\documentclass[a4paper,11pt]{article}
\usepackage{fullpage}

\pdfoutput=1
\usepackage[colorlinks=true,urlcolor=Blue,citecolor=Green,linkcolor=BrickRed]{hyperref}
\usepackage[dvipsnames]{xcolor}
\usepackage[T1]{fontenc}
\usepackage[utf8]{inputenc}
\usepackage{amssymb,amsmath,amsfonts,amsthm}
\usepackage{mathtools}
\usepackage{enumerate}
\usepackage[noadjust]{cite}
\usepackage{authblk}

\usepackage{tikz}
\usetikzlibrary{
  arrows.meta,
  backgrounds,
  calc,
  decorations.pathmorphing,
  decorations.pathreplacing,
  fit,
  positioning,
  shapes,
  shapes.geometric}

\def\dd{\mathinner{.\,.}}
\newcommand{\Oh}{\mathcal{O}}

\newcommand{\set}[1]{\left\lbrace #1 \right\rbrace}

\newcommand{\floor}[1]{\left\lfloor #1 \right\rfloor}

\DeclareMathOperator{\row}{row}
\DeclareMathOperator{\col}{col}
\DeclareMathOperator{\orig}{orig}
\DeclareMathOperator{\depth}{depth}
\DeclareMathOperator{\Span}{span}
\DeclareMathOperator{\rank}{rank}
\DeclareMathOperator{\select}{select}

\newcommand{\Ddep}{\mathcal{D}}
\newcommand{\Tcol}{\mathcal{T}}

\newcommand{\Top}{\top}
\newcommand{\Bot}{\bot}
\newcommand{\weight}{w}

\newcommand{\act}[1]{A_{#1}}
\newcommand{\acteq}[1]{S_{#1}}
\newcommand{\actup}[1]{T_{#1}}
\newcommand{\actL}[1]{L_{#1}}
\newcommand{\actT}[2]{T_{#1,#2}}

\newcommand{\interval}[1]{\Span(#1)}

\usepackage{thm-restate}
\newtheorem{theorem}{Theorem}[section]
\newtheorem{lemma}[theorem]{Lemma}

\newtheorem{observation}[theorem]{Observation}
\newtheorem{definition}[theorem]{Definition}

\usepackage[capitalise]{cleveref}
\crefname{theorem}{Theorem}{Theorems}
\crefname{lemma}{Lemma}{Lemmas}
\crefname{corollary}{Corollary}{Corollaries}
\crefname{definition}{Definition}{Definitions}
\crefname{observation}{Observation}{Observations}
\crefname{claim}{Claim}{Claims}
\crefname{fact}{Fact}{Facts}
\crefname{conjecture}{Conjecture}{Conjectures}
\crefname{remark}{Remark}{Remarks}

\makeatletter
\newcommand\thankssymb[1]{\textsuperscript{\@fnsymbol{#1}}}
\makeatother

\title{Near-Optimal and Efficient Encoding for Two-Dimensional Range Minimum Queries}

\author[1]{Pawe\l{} Gawrychowski\thanks{Partially supported by the Polish National Science Centre grant number 2023/51/B/ST6/01505.}}
\author[1]{Adam G\'orkiewicz\thankssymb{1}}
\author[2]{Srinivasa Rao Satti}
\affil[1]{Institute of Computer Science, University of Wroc\l{}aw, Poland}
\affil[2]{Department of Computer Science, Norwegian University of Science and Technology, Trondheim, Norway}

\date{}

\begin{document}

\maketitle

\begin{abstract}
	We consider the 2D RMQ encoding problem: given an $m\times n$ array of $mn$ elements over a total order, encode it such that, for any query rectangle, the position of its maximum element can be reported without accessing the original array.
	For $m \le n$, it is known how to encode the array in $\Oh(mn \min\{m, \log n\})$ bits with $\Oh(1)$-time queries [Brodal et al., Algorithmica 2012], and also how to obtain an asymptotically optimal encoding consisting of $\Oh(mn \log m)$ bits [Brodal et al., ESA 2013].
	However, the latter approach does not prove any guarantee on the query time, and it appears to be inherently sequential: it requires scanning the whole encoding to answer a query. We design a different encoding that uses near-optimal space while allowing for efficient queries.
	More concretely, for every parameter $\kappa\in[1, \log\log n]$, our encoding uses $\Oh(\kappa mn(\log m+\log\log n))$ bits and answers 2D RMQ queries in $\Oh(\log^{1/\kappa}n)$ time.
\end{abstract}

\section{Introduction}

Range minimum/maximum data structures are important substructures in the design of succinct and compressed data structures. The \emph{range maximum query} (RMQ) problem asks us to preprocess an array so that, given a query range, we can return the position of a maximum element in that range (traditionally, RMQ refers to range \emph{minimum} queries; we phrase the problem in terms of maxima, which is more convenient for our presentation, and the two variants are equivalent by negating all elements).
In the one-dimensional setting, the input is an array of length $n$ and a query is an interval \([i \dd j]\).
In two dimensions, the input is an $m\times n$ array and a query is an axis-parallel rectangle \([i_1 \dd i_2]\times[j_1 \dd j_2]\).

The RMQ problem has been studied in the literature in two distinct models: \emph{indexing} and \emph{encoding} models~\cite{Raman15}.  In the indexing model, the input array remains available at query time, and the data structure stores only auxiliary information, referred to as the \emph{index}, to support the queries efficiently. 
In the encoding model, which is more relevant to this paper, queries are answered without access to the original input, and hence only the information relevant to answer queries is encoded within the data structure.
The motivation behind the encoding model is twofold. First, it allows us to better understand the mathematical properties of the particular queries, and bound their effective entropy. Second, in some applications we need to support queries on an implicitly defined input that is well-defined but not stored at all (and not easy to compute in the query time).
There has been a significant amount of research in the recent past on designing indexes and encodings for various range queries on arrays, such as 
1D RMQ~\cite{FischerH10,FischerH11,BrodalDR12,Skala13,DavoodiRS17,FischerMN08,BarbayFN12,GawrychowskiJMW20,MunroNBW21,LiuY20,Liu21,JoS26},
2D RMQ~\cite{AmirFL07,DemaineLW14,YuanA10,BrodalDR12,BrodalDLRS16,BrodalBD13,GolinIKRSS16,JoS25,KaplanMNS17,GawrychowskiMW20},
range min\&max~\cite{GawrychowskiN15,JoS16,Tsur19,JoK25},
range median and range mode~\cite{KrizancMS05,Petersen08,PetersenG09,BoseKMT05,El-Zein0MNS19,GreveJLT10,ChanDLMW14},
range selection and range top-$k$~\cite{DavoodiNRR16,GrossiINRR17,GawrychowskiN15,JoLS21,El-Zein0MNS19}, and
range majority and range minority~\cite{KarpinskiN08,DurocherHMNS13,ChanDSW15,GagieHMN11,BelazzouguiGMNN21,NavarroT16,GawrychowskiN17,GagieHN20}. 
See~\cite{JoS25} for a survey on encodings for range queries on arrays.

RMQ is a basic primitive in data structures and algorithms.  In one dimension, it is tightly connected to the lowest common ancestor (LCA) problem on trees, and has applications in suffix-array based text indexing, compressed suffix trees, document retrieval, geometric indexing and many others. 
In two dimensions, it is a special case of the  orthogonal range-searching problem and is relevant as a subroutine in geometric indexing applications. The 1D problem is by now very well understood, with almost tight bounds known both in the indexing and the encoding models. On the other hand, the 2D problem is qualitatively more subtle: while almost optimal space indexes supporting constant-time queries are known, the encoding complexity is not well-settled. The gap between optimal encoding space and efficient query support for the 2D RMQ encoding problem remains one of the main structural differences between 1D and 2D RMQ, which is the problem we address in this paper.

\subsection{Previous results}

\paragraph{1D case.}
The 1D RMQ problem has been widely studied, and many linear-space (i.e., $\Oh(n \log n)$ bits) data structures with constant query-time have been proposed in the literature~\cite{BenderF00,HarelT84,SchieberV88}, most of them using the fact that RMQ is reducible to the LCA on the \emph{Cartesian tree} of the input array. The index space is further reduced to $\Oh(n/c)$ bits while supporting queries in $\Oh(c)$ time~\cite{FischerH11}, which is shown to be asymptotically optimal~\cite{BrodalDR12}. For the encoding model, the first breakthrough was due to Sadakane~\cite{Sadakane07} who gave an encoding of size $4n+o(n)$ bits that supports queries in constant time. The space usage is reduced to the optimal $2n+o(n)$ bits by Fischer and Heun~\cite{FischerH11}.
More recent work on 1D RMQ focused on the tradeoff between the query time and the lower order terms in the space usage: Liu~\cite{Liu21} showed that any data structure supporting RMQ queries using $2n - 1.5 \log n + r$ bits and query time $t$, must satisfy $r = n/(\log n)^{\Oh(t \log^2 t)}$, which closely matches the encoding of Liu and Yu~\cite{LiuY20} that uses $2n - 1.5 \log n +n/(\frac{\log n}{t})^{t}$ bits to support queries in $\Oh(t)$ time, for any parameter $t> 1$.
Thus the encoding complexity of 1D RMQ has been well-settled. 

\paragraph{2D case.}
For an $m \times n$ input array, the multidimensional solution of Gabow et al.~\cite{GabowBT84} gives a data structure that takes $\Oh(N \log^2 N)$ bits and supports queries in $\Oh(\log N)$ time, where $N = mn$. The subsequent results~\cite{Poon03,AmirFL07,YuanA10} improved these bounds to $\Oh(N \log N)$ bits and $\Oh(1)$ query time. Brodal et al.~\cite{BrodalDR12} designed the first linear-bit (i.e., $\Oh(N)$ bits) index that supports queries in constant time. They also proposed indexes that tradeoff query time for index size, which was further improved by~\cite{BrodalDLRS16}.

For the encoding model, Demaine et al.~\cite{DemaineLW14} proved that there is no natural Cartesian-tree analogue for the 2D RMQ problem, and in particular showed a counting lower bound of $\Omega(n^2\log n)$ bits for $n\times n$ arrays. Brodal et al.~\cite{BrodalDR12} extended the lower bound to $\Omega(mn\log m)$ bits for $m\times n$ arrays with $m\le n$. They also gave an $\Oh(1)$-query-time encoding using
$\Oh(mn\cdot \min\{m,\log n\})$ bits, which is optimal when $\log n=\Theta(\log m)$ but not in general.  The general encoding complexity was settled shortly afterwards by Brodal et al.~\cite{BrodalBD13}, who obtained an encoding that takes $\Theta(mn\log m)$ bits, matching the lower bound.
However, their encoding does not support queries efficiently, and answering one may require decoding essentially the whole representation.
Thus, prior to our work, there is a large gap between two extremes for the 2D RMQ encodings:
\begin{itemize}
    \item an encoding with $\Oh(1)$ query time using $\Oh(mn\min\{m,\log n\})$ bits, and
    \item an asymptotically optimal $\Oh(mn \log m)$-bit encoding without efficient query support.
\end{itemize}

Bridging this gap is the main motivation for our work. The natural question is whether one can approach the optimal $\Theta(mn\log m)$ space bound while still supporting fast queries.
Our main result answers this positively by showing a trade-off between space and query time.

\begin{theorem}\label{thm:main}
For any integer $\kappa\in[1, \log\log n]$, there is a $2$D-RMQ encoding for an $m\times n$ array that uses
$\Oh(\kappa mn(\log m+\log\log n))$
bits and answers queries in
$\Oh(\log^{1/\kappa}n)$
time.
\end{theorem}

Two regimes of \Cref{thm:main} are worth singling out.
For every constant $\epsilon > 0$, taking $\kappa$ to be a sufficiently large constant yields an encoding of $\Oh(mn(\log m + \log\log n))$ bits with queries in $\Oh(\log^{\epsilon} n)$ time; the space is asymptotically optimal whenever $n$ is at most exponential in $m$.
At the other extreme, taking $\kappa = \log\log n$ drives the query time down to $\Oh(1)$ at the cost of a multiplicative $\log\log n$ blowup in space.

\paragraph{Overview.}
\Cref{sec:reduction} reduces every query to a comparison between two distinguished points of the query rectangle, and \Cref{sec:origins} formalises the structure behind this comparison: every point is assigned an \emph{origin} in an auxiliary tree built over the columns, and the two points produced by the reduction always share enough structure in this tree for the comparison to be resolved efficiently.
\Cref{sec:compare} carries out the encoding itself in two passes.
The baseline version (\Cref{sec:baseline}) stores a small comparison structure at each node of the column tree and answers a query by repeatedly replacing one of the two current points with a substitute whose origin lies closer to the root; this matches \Cref{thm:main} at $\kappa = 1$.
The full version (\Cref{sec:full}) introduces a second tree, built over the depths of the column tree, that guarantees enough progress in each replacement to bound their number by its arity; tuning the arity yields the entire trade-off of \Cref{thm:main}.
Finally, \Cref{sec:implementation} constructs the two local primitives used by both encodings, by decomposing the active points of each node into $\Oh(m)$ \emph{quarter-rows}, inside which the weights are monotone with respect to the column.

\section{Preliminaries}\label{sec:prelim}

For integers $i, j$ we write $[i \dd j]$ to denote $\set{i, i + 1, \dots, j}$.
Throughout the paper, the input is an $m\times n$ array with rows indexed from $1$ to $m$, and columns indexed from $1$ to $n$.
A \emph{point} is a pair $(i,j)\in[1\dd m]\times[1\dd n]$.
We treat the input as a function $\weight$ mapping each point to a distinct weight;
for a point $p=(i,j)$, we write $\weight(p)$ for its weight, $\row(p)\coloneqq i$, and $\col(p)\coloneqq j$.
For a set of points $S$, we write $\max_\weight(S)$ for the unique point $p \in S$ maximising~$\weight(p)$.
A~\emph{query rectangle} is a set of the form $R = [i_1 \dd i_2] \times [j_1 \dd j_2] \subseteq [1\dd m]\times[1\dd n]$, and the 2D RMQ problem asks, given such a rectangle $R$, to return the heaviest point in it, namely $\max_\weight(R)$.

We measure the space in bits and assume the standard word RAM model of computation, with word size $\Theta(\log n)$, when describing the query algorithms.

We make use of the following 1D RMQ encoding.

\begin{lemma}[\cite{FischerH11}]\label{lem:rmq1d}
	For an array of length $n$ with distinct values, there is an RMQ structure using $\Oh(n)$ bits that returns the position of the maximum in any interval in $\Oh(1)$ time.
\end{lemma}

In our encodings, we use representations of sets supporting the following operations. Given a set $S \subseteq [1 \dd U]$,

\begin{itemize}
\item $\rank_S(x)$ returns the number of elements of $S$ that are not greater than $x$, and
\item $\select_S(i)$ returns the $i$-th smallest element in $S$.
\end{itemize}

We will only query rank on elements $x \in S$ (i.e., only partial-rank support is required), which allows for an encoding with a smaller redundancy term
compared to the general case.

\begin{lemma}[\cite{RamanRR07}]\label{lem:sparse}
	Let $S\subseteq[1 \dd U]$ with $|S|=k$.
	There is an encoding using $\Oh(k + k\log(U/k))$ bits supporting rank (only for members of $S$), and select queries in $\Oh(1)$ time.
\end{lemma}

We also use the Elias--Fano encoding for representing monotonic sequences:

\begin{lemma}[\cite{FerraginaL25}]\label{lem:mono}
	Let $S[1 \dd k]$ be a non-decreasing sequence with values in $[1 \dd U]$, with $k \le U$.
	There is an encoding using $\Oh(k + k\log(U/k))$ bits supporting access to any $S[t]$ in $\Oh(1)$ time.
	The non-increasing case is symmetric.
\end{lemma}

\Cref{lem:sparse} and \Cref{lem:mono} can be further refined to optimise the leading constants, but this introduces clutter and does not improve our final result.

Throughout the paper, the space of a structure is almost always obtained by summing entropy-like bounds of the form $k_i \log(U_i/k_i)$ coming from the two encodings above.
To handle such sums uniformly, we will rely on the following inequality.

\begin{lemma}[log-sum inequality]\label{lem:logsum}
	For any positive numbers $k_1,\dots,k_t$ and $U_1,\dots,U_t$, with total sums $K \coloneqq \sum_{i=1}^t k_i$ and $U \coloneqq \sum_{i=1}^t U_i$,
	\[
		\sum_{i=1}^t k_i \log \frac{U_i}{k_i}
		\le K \log \frac{U}{K}.
	\]
\end{lemma}
\begin{proof}
	For every $i$, let $\alpha_i \coloneqq k_i/K$, so that $\sum_{i=1}^t \alpha_i = 1$.
	Then
	\[
		\sum_{i=1}^t k_i \log \frac{U_i}{k_i}
		= K \sum_{i=1}^t \alpha_i \log \frac{U_i/K}{\alpha_i}.
	\]
	Since $\sum_{i=1}^t \alpha_i = 1$ and the logarithm is concave, Jensen's inequality yields
	\[
		\sum_{i=1}^t \alpha_i \log \frac{U_i/K}{\alpha_i}
		\le \log\!\left(\sum_{i=1}^t \alpha_i \cdot \frac{U_i/K}{\alpha_i}\right)
		= \log\!\left(\sum_{i=1}^t \frac{U_i}{K}\right)
		= \log\frac{U}{K}.
	\]
	Multiplying both sides by $K$ gives the claim.
\end{proof}

\section{Reducing a Query to Two Candidates}\label{sec:reduction}
The first step of the construction reduces the query rectangle to comparing two points.
Since storing enough information to compare arbitrary pairs is far too expensive, the reduction must produce points with some additional structure.

\begin{definition}[visibility]\label{def:visibility}
Let $p=(i,j)$ be a point and let $j'\in[1 \dd n]$ be a column.
We say that $p$ is \emph{visible from column $j'$} if $p = \max_w(\{i\}\times[\min\{j,j'\}\dd \max\{j,j'\}])$.
Equivalently, along row $i$, the point $p$ remains the heaviest when scanning from column $j$ toward column~$j'$.
\end{definition}


\begin{lemma}\label{lem:two-cand}
	There is a data structure using $\Oh(mn\log m)$ bits that, given a query rectangle $R$, returns in $\Oh(1)$ time two (not necessarily distinct) points $p,q\in R$ such that $\max_\weight(R)=\max_\weight(\{p,q\})$, $p$ is visible from $\col(q)$, and $q$ is visible from $\col(p)$.
\end{lemma}
\begin{proof}
	Build a 1D RMQ structure of \Cref{lem:rmq1d} for every row and column.
	In addition, for every $k \ge 0$ and every starting row $r$, consider the row block
	$[r \dd r+2^k-1]\times[1 \dd n]$.
	For each such block, define an auxiliary array of length $n$ whose $j$-th entry corresponds to the heaviest point in column $j$ restricted to the block, and store its 1D RMQ structure of \Cref{lem:rmq1d}.
	Over all blocks, the total space is $\Oh(mn\log m)$ bits.
	
	Suppose $R = [i_1 \dd i_2]\times[j_1 \dd j_2]$ and let $h \coloneqq i_2-i_1+1$.
     If $h$ is a power of $2$, we can answer the query in $\Oh(1)$ time using the 1D RMQ   structure for the row block $[i_1 \dd i_2]\times[1 \dd n]$.
	Otherwise, let 
		$R_{\mathrm{top}} \coloneqq [i_1 \dd i_1+2^{\floor{\log h}}-1]\times[j_1 \dd j_2]$
and
		$R_{\mathrm{bot}} \coloneqq [i_2-2^{\floor{\log h}}+1 \dd i_2]\times[j_1 \dd j_2]$
	be two overlapping rectangles. Let $p \coloneqq \max_\weight(R_{\mathrm{top}})$ and $q \coloneqq \max_\weight(R_{\mathrm{bot}})$.
	Since the two blocks cover $R$, the maximum of $R$ is the heavier of $p$ and $q$.
	
	To prove mutual visibility, consider $p=(i,j)$ and the column $j'=\col(q)$.
	If $p$ was not visible from $j'$, then some point $(i,t)$ with $t$ between $j$ and $j'$ would satisfy $\weight(i,t)>\weight(p)$.
	Since both $j$ and $j'$ lie in $[j_1 \dd j_2]$, the point $(i,t)$ would belong to $R_{\mathrm{top}}$, contradicting the maximality of $p$ there.
	The proof for $q$ is symmetric.
\end{proof}
By applying the above preprocessing, the RMQ problem is reduced to comparing the two mutually visible points $p$ and $q$, called \emph{candidates}, and returning the heavier one.

In principle, the ability to compare arbitrary pairs of points would reveal the total order of all $mn$ entries and thus requires $\Omega(mn \log (mn))$ bits of space.
The visibility condition provides the additional structure that makes the comparison problem compressible.
The rest of the paper is therefore concerned only with comparing mutually visible pairs of
points.

\section{Active Points and Origins}\label{sec:origins}
We now introduce the main concepts of the data structure, and formulate exactly which pairs of points are going to be comparable by it.
We then show that the pair of candidates produced by \Cref{lem:two-cand} satisfies the requirement.

\paragraph{Tree on columns.}
Let $\Tcol$ be a complete binary tree over the columns $[1 \dd n]$.
Each leaf of $\Tcol$ corresponds to a column, with the leaves ordered from left to right.
Every node $u \in \Tcol$ corresponds to an interval, denoted as
$\interval{u} \subseteq [1 \dd n]$, consisting of the columns of the leaves in the subtree of $u$.
In particular, the interval of the root is the whole $[1 \dd n]$, and the intervals of the two children of an internal node partition the interval of their parent into two parts of almost equal length.
The depth of $u$ in $\Tcol$ is denoted as $\depth(u)$, where the root has depth $1$.

\begin{definition}[active points]\label{def:active}
	For an internal node $u\in\Tcol$ with children $u_1$ and $u_2$, whose intervals are
	$\interval{u_1}=[a_1\dd b_1]$ and $\interval{u_2}=[a_2\dd b_2]$,
	we define
	$\partial u \coloneqq \set{a_1, b_1, a_2, b_2}$.
	Note that some of these endpoints may coincide.
	For a leaf $u\in\Tcol$, we define $\partial u \coloneqq \{c\}$, where $c$ is the unique column in $\interval{u}$.
	
	We define the \emph{active set} of $u$, denoted $\act{u}$, as the set of all points in $[1\dd m]\times \interval{u}$ that are visible from at least one column in $\partial u$.
	A point $p$ is said to be \emph{active in} $u$ if $p \in \act{u}$.
\end{definition}
	
The notion of active points is useful because `activeness' behaves monotonically down the tree: once a point is active in some node, it remains active on the path toward its column~leaf.

\begin{lemma}\label{lem:active-path}
	For every point $p$, the set of nodes $u$ such that $p \in \act{u}$ forms a downward path in $\Tcol$ ending at the leaf corresponding to column $\col(p)$.
\end{lemma}
\begin{proof}
	Suppose that $p=(i,j)$.
	If $p\in \act{u}$, then by definition $j\in \interval{u}$.
	Hence every node $u$ such that $p\in \act{u}$ lies on the root-to-leaf path ending at the leaf corresponding to column~$j$.

	It therefore suffices to show that this set of nodes is downward closed.
	Let $u$ be an internal node such that $p\in \act{u}$, and let $v$ be the child of $u$ whose interval contains $j$.
	Choose a witness column $c\in \partial u$ such that $p$ is visible from $c$.

	If $c\in \partial v$, then the same witness shows that $p\in \act{v}$.
	Otherwise, let $c'\in \partial v$ be the endpoint of $\interval{v}$ closest to $c$.
	Then the interval between $j$ and $c'$ is contained in the interval between $j$ and $c$.
	Since $p$ is visible from $c$, it is also visible from $c'$, and hence $p\in \act{v}$.

	It follows that the active set of $p$ is a downward-closed subset of this root-to-leaf path.
	Since $p$ is always active in the leaf for column $j$ (it is trivially visible from its own column), the set is nonempty, and therefore forms a path ending at the leaf corresponding to column~$j$.
\end{proof}

\begin{definition}[origin]\label{def:origin}
	For a point $p$, we define the \emph{origin} of $p$, denoted $\orig(p)$, as the topmost node of the path given by \Cref{lem:active-path}.
\end{definition}

\begin{definition}[co-active pair]\label{def:coactive}
	A pair of points $p, q$ is \emph{co-active} if there exists a node $u \in \Tcol$ such that $p, q \in \act{u}$.
\end{definition}

The next lemma gives a more algorithmic characterization of co-activity.
Although the definition only requires the existence of some witness node in which both points are active, one can always pick a canonical one: the deeper of the two origins.

\begin{lemma}\label{lem:coactive-char}
	Let $p$ and $q$ be points such that
	$\depth(\orig(p)) \le \depth(\orig(q))$.
	Then $p$ and $q$ are co-active if and only if $p$ is active in $\orig(q)$.
\end{lemma}
\begin{proof}
	If $p$ is active in $\orig(q)$, then both $p$ and $q$ belong to $\act{\orig(q)}$, so the pair is co-active.
	Conversely, suppose that $p$ and $q$ are co-active.
	Then there exists a node $u\in\Tcol$ such that $p,q\in \act{u}$.
	By \Cref{lem:active-path}, both $\orig(p)$ and $\orig(q)$ are ancestors of $u$.
	Since
	$\depth(\orig(p)) \le \depth(\orig(q))$,
	it follows that $\orig(q)$ is a descendant of $\orig(p)$.
	Hence $\orig(q)$ lies on the active path of $p$, and therefore $p$ is active in $\orig(q)$.
\end{proof}

It can be directly verified that the pairs produced by the reduction of \Cref{lem:two-cand} are co-active.

\begin{lemma}\label{lem:mv-implies-coactive}
	If $p$ is visible from column $\col(q)$ and $q$ is visible from column
	$\col(p)$, then $p$ and $q$ are co-active.
\end{lemma}
\begin{proof}
	If $\col(p)=\col(q)$, then both $p$ and $q$ are active in the leaf of $\Tcol$ corresponding to this column, so the pair is co-active.
	Otherwise, let $u$ be the lowest common ancestor in $\Tcol$ of the leaves corresponding to columns $\col(p)$ and $\col(q)$.
	In that case $u$ is an internal node, and one of the columns $c \in \partial u$ lies between $\col(p)$ and $\col(q)$.
	Since $p$ is visible from $\col(q)$, it is also visible from every column between $\col(p)$ and $\col(q)$, and in particular from $c$.
	Hence $p\in \act{u}$.
	By the same argument, $q\in \act{u}$.
	Therefore the pair is co-active.
\end{proof}

In the remainder of the paper, we focus on encodings capable of efficiently comparing precisely the co-active pairs of points.
Combined with the preprocessing of \Cref{lem:two-cand}, which reduces every RMQ query to such a comparison by \Cref{lem:mv-implies-coactive}, this yields the main result.

\section{Comparing Co-Active Pairs}\label{sec:compare}

From this point onward, we consider the following comparison problem:
given two co-active points $p_0$ and $q_0$, determine which of them is heavier.
We work with its decision version, fixing the hypothesis $\mathcal{H} : \weight(p_0) \le \weight(q_0)$, so that the answer is $q_0$ if $\mathcal{H}$ holds, and $p_0$ otherwise.

At a high level, the query algorithm maintains a co-active pair of points, initially $(p_0, q_0)$, along with the invariant that the hypothesis `the first point is no heavier than the second' is equivalent to $\mathcal{H}$.
At each step, one of the two current points gets replaced by a substitute whose origin in $\Tcol$ is strictly closer to the root, in such a way that the invariant is preserved.
Note that this strategy does not preserve the identity of the heavier point of the pair.
The process terminates once a single local comparison can resolve the query.

\subsection{Local Primitives}

We now state the two local primitives that support the process described above; their implementation is deferred to \Cref{sec:implementation}.
The first primitive is a \emph{ranking} structure: given a set of points, it answers rank queries among them, and is used to perform the single comparison that resolves the query at the end of the process.
The second primitive is a \emph{lifting} structure: it takes a point and returns a suitable substitute whose origin is closer to the root of $\Tcol$, and is used to carry out the intermediate replacement steps.

\paragraph{Ranks of points.}
For a set of points $L$ and a point $p \in L$, we define the \emph{rank} of $p$ in $L$ as the number of points in $L$ not heavier than $p$.
In particular, knowing the ranks of two points in $L$ immediately tells us which of them is heavier.

\begin{restatable}{lemma}{RestateRankPrimitive}\label{lem:rank-primitive}
Let $u$ be a node of $\Tcol$ with $\lambda_u \coloneqq |\interval{u}|$.
For every nonempty set $L\subseteq \act{u}$, there exists an encoding using
\[
	\Oh\!\left(|L|\log m+|L|\log\frac{m\lambda_u}{|L|}\right)
\]
bits that, given a point $p\in L$, returns in $\Oh(1)$ time the rank of $p$ in $L$.
\end{restatable}

The intended use of \Cref{lem:rank-primitive} is simple:
if we have built the encoding for some subset $L\subseteq \act{u}$
and both points of the current pair happen to lie in $L$,
the decision problem can be resolved immediately by comparing their ranks in $L$.

\paragraph{Predecessors and successors.}
For a set of points $V$ and a point $p$, the \emph{predecessor} of $p$ in $V$ is the heaviest point in
$\set{q\in V : \weight(q) \le \weight(p)}$,
and the \emph{successor} of $p$ in $V$ is the lightest point in
$\set{q\in V : \weight(p) \le \weight(q)}$.
If the first set is empty, the predecessor is defined to be $\Bot$, and if the second set is empty, the successor is defined to be $\Top$, whose weights are understood to be $-\infty$ and $+\infty$, respectively.


\begin{observation}\label{obs:equiv}
	Let $p, q$ be points, and let $P, Q$ be such that $p \in P$ and $q \in Q$.
	Let $p'$ be the successor of $p$ in $Q$, and let $q'$ be the predecessor of $q$ in $P$.
	Then
	\[
		\weight(p) \le \weight(q) \iff \weight(p') \le \weight(q) \iff \weight(p) \le \weight(q').
	\]
	\end{observation}

Thus, provided that the other point belongs to the target set, one may replace one of the current points with an appropriate predecessor or successor, while preserving the invariant.

\begin{restatable}{lemma}{RestateLiftPrimitive}\label{lem:lift-primitive}
Let $u$ be a node of $\Tcol$ with $\lambda_u \coloneqq |\interval{u}|$.
For every nonempty source set $S\subseteq \act{u}$ and target set $T\subseteq \act{u}$, there exists an encoding using
\[
  \Oh\!\left(|S|\log m+ |S|\log\frac{m\lambda_u}{|S|}\right)
\]
bits that, given a point $p\in S$, returns in $\Oh(1)$ time the predecessor and successor of $p$ in $T$.
\end{restatable}

The remainder of the paper is devoted to choosing suitable sets for
\Cref{lem:rank-primitive,lem:lift-primitive}.
The warm-up structure uses one lifting set per node of $\Tcol$ and therefore raises the origin by one or more levels at a time.
The faster structure stores several such sets, organised by a tree over the depth range of $\Tcol$.

\subsection{The Baseline Structure}\label{sec:baseline}

We first describe a simpler $\Oh(mn(\log m + \log\log n))$-bit encoding with $\Oh(\log n)$ query time
before we move to the full trade-off.

\paragraph{Stored structures.}
For every point $p$, we store $\depth(\orig(p))$.
Note that $\orig(p)$ is uniquely determined by $\col(p)$ and $\depth(\orig(p))$, as it is the unique node of $\Tcol$ at that depth whose interval contains $\col(p)$, by \Cref{lem:active-path}.
We assume standard constant-time navigation in $\Tcol$, which can be implemented with bitwise operations.

For every node $u \in \Tcol$, we define $\acteq{u}$ as the set of points in $\act{u}$ whose origin is $u$ itself, and $\actup{u}$ as those whose origin is a strict ancestor of~$u$; by \Cref{lem:active-path}, these two sets partition~$\act{u}$.
For each such node $u$, we store the following two structures:
\begin{itemize}
	\item by \Cref{lem:rank-primitive}, a local ranking structure for the set $\acteq{u}$, and
	\item by \Cref{lem:lift-primitive}, a lifting structure with source set $\acteq{u}$ and target set $\actup{u}$.
\end{itemize}

The idea is as follows.
If the two current points have the same origin $u$, then both belong to $\acteq{u}$, and the ranking structure stored for $u$ resolves the comparison.
Otherwise, let $u$ be the deeper of the two origins; then the corresponding point belongs to $\acteq{u}$, while the other point is also active in $u$ with its origin strictly above, and hence belongs to $\actup{u}$.
In this case, the lifting structure stored for $u$ is used.
We now formalize this.

\paragraph{Query algorithm.}
Let $R$ be the query rectangle.
Using \Cref{lem:two-cand}, we obtain two candidate points $p_0, q_0\in R$ such that $\max_\weight(R)$ is the heavier of the two.
By \Cref{lem:mv-implies-coactive}, the pair $(p_0, q_0)$ is co-active.
We fix the hypothesis $\mathcal{H} : \weight(p_0) \le \weight(q_0)$ and maintain an ordered pair $(p,q)$, initially $(p_0, q_0)$, with the invariant that the statement $\weight(p) \le \weight(q)$ is equivalent to $\mathcal{H}$.

If $\orig(p) = \orig(q) = u$, then $p, q\in \acteq{u}$, and the ranking structure stored for $u$ resolves the comparison immediately.
From now on assume that the two origins are different.

Suppose first that $\depth(\orig(q)) > \depth(\orig(p))$, and let $u \coloneqq \orig(q)$.
By \Cref{lem:coactive-char}, we have $p \in \act{u}$, and hence $p \in \actup{u}$.
In this case we use the lifting structure stored for $u$ to obtain $q'$, the predecessor of $q$ in $\actup{u}$.
By the invariant, the statement $\weight(p)\le\weight(q)$ is equivalent to $\mathcal{H}$, and since $p\in\actup{u}$, \Cref{obs:equiv} tells us that the same holds for $\weight(p)\le\weight(q')$.
At this point, either $q'=\Bot$, in which case the hypothesis is false and the query is resolved immediately, or else $q'\in \actup{u}$, and we replace $q$ with $q'$.
In the latter case, the origin of $q'$ lies strictly above $u$, thus the origin of the deeper point has moved upward.
Also, since both $p$ and $q'$ belong to $\act{u}$, the new pair $(p,q')$ is again co-active.

The case $\depth(\orig(p)) > \depth(\orig(q))$ is symmetric: we now take $u \coloneqq \orig(p)$, and use the lifting structure stored for $u$ to replace $p$ with its \emph{successor} $p'$ in $\actup{u}$, where $q\in\actup{u}$ by \Cref{lem:coactive-char}.
As before, \Cref{obs:equiv} preserves the invariant, and the two possible outcomes mirror those of the previous case: either $p'=\Top$, which resolves the query, or we get a substitute point $p'\in\actup{u}$, with origin strictly above $u$, and a new co-active pair $(p', q)$.

The algorithm continues in this way until the two current points have the same origin, at which point the ranking structure finishes the comparison.
Since the height of $\Tcol$ is $\Oh(\log n)$, there can be at most $\Oh(\log n)$ such steps, and all work performed in each step is constant-time.
Hence the total query time is $\Oh(\log n)$.

\paragraph{Space analysis.}
Most of the space analysis is taken care of by the following lemma, which we state as a self-contained bound.
Its statement is tailored to match the per-node space bound of \Cref{lem:rank-primitive,lem:lift-primitive}, so that applying it reduces to verifying that the relevant families of sets are pairwise disjoint.
Isolating the statement in this way also lets us reuse the same bound later in the space analysis of the full structure.

\begin{lemma}\label{lem:baseline-entropy}
	Suppose that with every node $u\in\Tcol$ we associate a set of points $X_u$,
	and that the family $\{X_u : u\in\Tcol\}$ is pairwise disjoint.
	If for every $u$ with $X_u \neq \emptyset$ we store a structure~in
	$\Oh(|X_u|\log m + |X_u|\log(m\lambda_u/|X_u|))$
	bits, where $\lambda_u \coloneqq |\interval{u}|$, then the total space of all these structures is
	$\Oh(mn\log m + mn\log\log n)$
	bits.
\end{lemma}
\begin{proof}
	We split the per-node bound into its two summands and handle them separately.
	For the first summand, since the sets $X_u$ are pairwise disjoint subsets of $[1\dd m]\times[1\dd n]$, we have $\sum_u |X_u| \le mn$, and hence
	\[
		\sum_u |X_u|\log m \le mn\log m.
	\]
	It remains to bound $\sum_u |X_u|\log(m\lambda_u/|X_u|)$ by $\Oh(mn\log\log n)$.

	Let $H$ be the height of $\Tcol$ (so $H=\Oh(\log n)$).
	For every depth $d\in[1\dd H]$, let $\mathcal{N}_d$ be the set of nodes $u$ at depth $d$ with $X_u\neq\emptyset$, and define
	\[
		K_d \coloneqq \sum_{u\in\mathcal{N}_d} |X_u|.
	\]
	Since the sets $X_u$ are pairwise disjoint, we have
	$\sum_{d=1}^H K_d \le mn$.

	Fix a depth $d$ with $K_d > 0$.
	The nodes of $\Tcol$ at depth $d$ have pairwise disjoint intervals, so
	$\sum_{u\in\mathcal{N}_d} m\lambda_u \le mn$.
	Applying \Cref{lem:logsum} to the values $k_u \coloneqq |X_u|$ and $U_u \coloneqq m\lambda_u$ for $u\in\mathcal{N}_d$, we obtain
	\[
		\sum_{u\in\mathcal{N}_d} |X_u| \log \frac{m\lambda_u}{|X_u|}
		\le K_d \log \frac{mn}{K_d}.
	\]

	Now summing over all depths with $K_d>0$ and applying \Cref{lem:logsum} once again, this time to the values $K_d$ with $U_d=mn$ for every such $d$, we get the total cost bound
	\[
		\sum_{\substack{1\le d\le H \\ K_d>0}} K_d \log \frac{mn}{K_d}
		\le mn \log \frac{H\cdot mn}{mn}
		= mn\log H
		= \Oh(mn\log\log n). \qedhere
	\]
\end{proof}

We now sum the space over all components.
For every point $p$, the value $\depth(\orig(p))$ belongs to $[1\dd H]$, where $H=\Oh(\log n)$ is the height of $\Tcol$, and thus can be stored in $\Oh(\log\log n)$ bits per point, for a total of $\Oh(mn\log\log n)$ bits.

Next, we bound the total space of the local ranking structures.
For every node $u\in\Tcol$ with $\acteq{u}\neq\emptyset$, we store one such structure on the set $\acteq{u}$, and its per-node cost, given by \Cref{lem:rank-primitive}, coincides exactly with the one assumed in \Cref{lem:baseline-entropy} (with $X_u=\acteq{u}$).
Since every point has exactly one origin, the family $\{\acteq{u} : u\in\Tcol\}$ partitions the set of all points, and in particular is pairwise disjoint.
Hence \Cref{lem:baseline-entropy} applies and bounds the total space of the ranking structures by $\Oh(mn\log m + mn\log\log n)$ bits.

The same argument applies to the lifting structures: each is built on $\acteq{u}$ as its source set, and the per-node bound in \Cref{lem:lift-primitive} is identical to that of \Cref{lem:rank-primitive}.
Their total space is therefore also $\Oh(mn\log m + mn\log\log n)$ bits.

Finally, the reduction structure of \Cref{lem:two-cand} uses $\Oh(mn\log m)$ bits, which completes the proof, as all components fit within $\Oh(mn(\log m + \log\log n))$ bits of space.

\subsection{The Full Structure}\label{sec:full}

We now describe the general construction achieving the trade-off from \Cref{thm:main}.
As in the baseline structure, the query algorithm fixes a comparison hypothesis and repeatedly replaces one of the two current points.
The difference is that we no longer move one level of depth at a time in the tree $\Tcol$.
Instead, we group consecutive depths of $\Tcol$ into larger blocks, and each lifting step moves the deeper point from its current block into an earlier one.
The local ranking structures are also adapted: rather than comparing only points with exactly the same origin, they compare points whose origins lie in one prescribed depth block.
As in the baseline structure, we store $\depth(\orig(p))$ for every point $p$.

\paragraph{Tree on depths.}
Fix a parameter $\kappa\in[1,\log\log n]$ and let
$
	\tau \coloneqq  \lceil\log^{1/\kappa}n\rceil.
$
Recall that the height of the complete binary tree $\Tcol$ is $H=\Oh(\log n)$.
We build a complete $\tau$-ary tree $\Ddep$ over the depths $[1\dd H]$ by recursively partitioning every interval into at most $\tau$ consecutive sub-intervals of almost equal lengths, with the children of every node ordered so that smaller depths lie to the left.
For every node $x\in\Ddep$, let $\interval{x}\subseteq[1\dd H]$ denote the depth interval represented by $x$.
The height of $\Ddep$ is~$\Oh(\kappa)$.

The tree $\Ddep$ is an auxiliary structure used to organise the query process.
We assume standard constant-time navigation in both $\Ddep$ and $\Tcol$, which can be implemented with bitwise operations (for $\Ddep$ this requires rounding up $\tau$ to the nearest power of two).

Given two points, their origin depths determine two leaves of $\Ddep$, and hence their lowest common ancestor node, whose children separate them.
The algorithm will work relative to this one node of $\Ddep$: it repeatedly moves the deeper origin across its children until both current origin depths lie in the interval of the same child.
At that point a local ranking structure finishes the comparison.

\paragraph{Ranking structures.}
For every depth $d\in[1\dd H]$, let $C_d$ denote the largest \emph{canonical} interval of $\Ddep$ of the form $[a\dd d]$, where an interval is canonical if it appears as $\interval{x}$ for some node $x\in\Ddep$.
Now let $u\in\Tcol$ be any node of depth $d$.
We define
\[
	\actL{u} \coloneqq \set{p\in \act{u} : \depth(\orig(p))\in C_d}.
\]
For every node $u\in\Tcol$, we store one local ranking structure of \Cref{lem:rank-primitive} for the set $\actL{u}$.

\paragraph{Lifting structures.}
We next describe the lifting structures.
Let $y$ be a non-leftmost child of a node $x\in\Ddep$.
If we denote $\interval{x}=[a\dd b]$ and $\interval{y}=[a'\dd b']$,
then the depths lying in child intervals of $x$ strictly to the left of $y$ are precisely those in~$[a\dd a'-1]$.

Now let $u\in\Tcol$ be a node such that $\depth(u)\in\interval{y}$.
We define
\[
	\actT{y}{u} \coloneqq \set{p\in \act{u} : \depth(\orig(p))\in [a\dd a'-1]}.
\]
In words, $\actT{y}{u}$ consists of the points that are active in $u$, whose origin depths lie in one of the intervals of the left siblings of $y$.
For every such pair $(y,u)$, we store a lifting structure of \Cref{lem:lift-primitive} with source set $\acteq{u}$ (the points with origin $u$) and target set $\actT{y}{u}$.

This is a natural generalisation of the baseline structure:
instead of lifting a point with origin $u$ from $\depth(u)$ to \emph{any} smaller depth, we now lift it only into the union of the child intervals of $x$ that lie to the left of the one currently containing $\depth(u)$.
In particular, for a fixed node $u\in\Tcol$, the sets $\actT{y}{u}$ partition the set $\actup{u}$ (points active in $u$ whose origins lie strictly above $u$) according to the first branching node in $\Ddep$ at which their origin depths~diverge.

\paragraph{Query algorithm.}
Let $R$ be the query rectangle.
As before, \Cref{lem:two-cand} returns two candidate points $p_0,q_0\in R$, and \Cref{lem:mv-implies-coactive} guarantees that this pair is co-active.
We fix the hypothesis $\mathcal{H} : \weight(p_0)\le \weight(q_0)$ and maintain an ordered pair $(p,q)$, initially $(p_0,q_0)$, with the same invariant as in the baseline structure.

If the two initial points have the same origin $u$, then both belong to $\actL{u}$, since both are in $\act{u}$ and $\depth(u)\in C_{\depth(u)}$.
Hence the ranking structure stored for $u$ compares them immediately.
Assume from now on that the current origins are different.

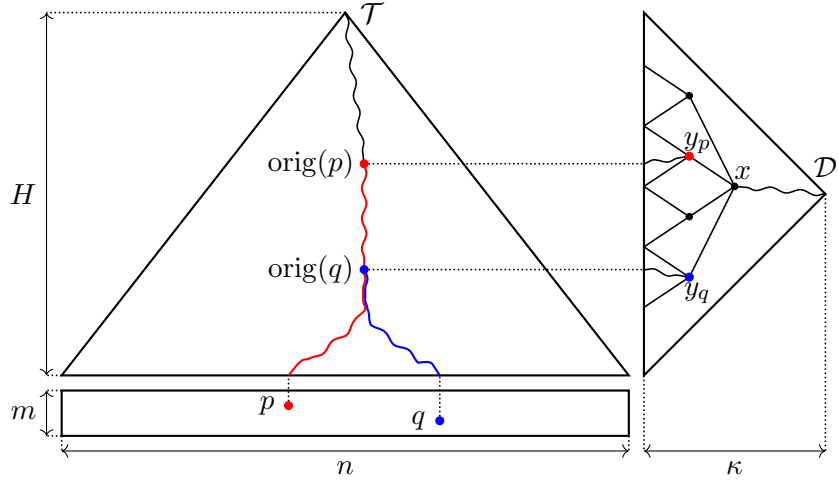
\begin{figure}
  \centering
  \begin{tikzpicture}
  \def\m{6mm}
  \def\n{75mm}
  \def\fatsep{1.3pt}
  \def\pathpos{40mm}

  \def\px{30mm}
  \def\qx{50mm}
  \def\py{4mm}
  \def\qy{2mm}
  \def\divy{18mm}

  \def\treespace{2mm}

  \coordinate(wDL) at (0, 0);
  \coordinate(wDR) at (\n, 0);
  \coordinate(wUR) at (\n, \m);
  \coordinate(wUL) at (0, \m);
  \coordinate(wUM) at ($(wUL)!0.5!(wUR)$);

  \draw[thick] (wDL) rectangle (wUR);

  \coordinate(TL) at ([yshift=\treespace]wUL);
  \coordinate(TR) at ([yshift=\treespace]wUR);
  \coordinate(TU) at ([yshift=5cm]wUM);

  \draw[thick] (TL) -- (TR) -- (TU) -- cycle;

  \coordinate (DD) at ([xshift=\treespace]TR);
  \coordinate (DU) at (DD |- TU);
  \coordinate (DL) at ($(DD)!0.5!(DU)$);
  \coordinate (DR) at ([xshift=24mm]DL);

  \draw[thick] (DD) -- (DR) -- (DU) -- cycle;

  \coordinate (X) at ([xshift=12mm, yshift=1mm]DL);

  \def\H{6mm}
  \def\span{8mm}

  \coordinate (y2) at (X -| DL);
  \coordinate (y1) at ([yshift=-\span]y2);
  \coordinate (y0) at ([yshift=-\span]y1);
  \coordinate (y3) at ([yshift=+\span]y2);
  \coordinate (y4) at ([yshift=+\span]y3);

  \coordinate (r0) at ([xshift=\H]$(y0)!0.5!(y1)$);
  \coordinate (r1) at ([xshift=\H]$(y1)!0.5!(y2)$);
  \coordinate (r2) at ([xshift=\H]$(y2)!0.5!(y3)$);
  \coordinate (r3) at ([xshift=\H]$(y3)!0.5!(y4)$);

  \coordinate (Tp) at (\px, \m + \treespace);
  \coordinate (Tq) at (\qx, \m + \treespace);
  \coordinate (Div) at (\pathpos, \divy);
  \coordinate (OP) at ([yshift=-1mm]Div |- r2);
  \coordinate (OQ) at ([yshift=+1mm]Div |- r0);
  \coordinate (p) at (\px, \py);
  \coordinate (q) at (\qx, \qy);

  \def\pathoffset{0.4pt} 

  \tikzset{snake path/.style={thick, decorate, decoration={snake, amplitude=0.3mm, segment length=4mm}}}

  \draw[snake path, black, semithick] (DR) -- (X);
  \draw[snake path, black, semithick] (TU) -- (OP);
  \draw[snake path, black, semithick] (OP -| DL) -- (r2);
  \draw[snake path, black, semithick] (OQ -| DL) -- (r0);

  \draw[red, snake path] (OP) -- (OQ);

  \draw[red,  snake path, transform canvas={xshift=-\pathoffset}] (OQ) -- (Div);
  \draw[blue, snake path, transform canvas={xshift=+\pathoffset}] (OQ) -- (Div);

  \draw[red,  thick, decorate, decoration={snake, amplitude=0.3mm, segment length=4mm}] ([xshift=-\pathoffset]Div) to[out=270, in=60, looseness=0.6] (Tp);
  \draw[blue, thick, decorate, decoration={snake, amplitude=0.3mm, segment length=4mm}] ([xshift=+\pathoffset]Div) to[out=270, in=110, looseness=0.6] (Tq);

  \draw[densely dotted, semithick] (OP) -- (OP -| DL);
  \draw[densely dotted, semithick] (OQ) -- (OQ -| DL);

  \draw[densely dotted, semithick] (p) -- (Tp);
  \draw[densely dotted, semithick] (q) -- (Tq);

  \tikzset{subtree outline/.style={semithick}}
  \draw[subtree outline] (y0) -- (y1) -- (r0) -- cycle;
  \draw[subtree outline] (y1) -- (y2) -- (r1) -- cycle;
  \draw[subtree outline] (y2) -- (y3) -- (r2) -- cycle;
  \draw[subtree outline] (y3) -- (y4) -- (r3) -- cycle;

  \tikzset{subtree edge/.style={semithick}}
  \draw[subtree edge] (X) -- (r0) {};
  \draw[subtree edge] (X) -- (r1) {};
  \draw[subtree edge] (X) -- (r2) {};
  \draw[subtree edge] (X) -- (r3) {};

  \tikzset{subtree root/.style={circle, fill=black, inner sep=1pt}}
  \tikzset{fat dot/.style={circle, fill=black, inner sep=1.2pt}}

  \node[fat dot, blue] at (r0) {};
  \node[subtree root] at (r1) {};
  \node[fat dot, red] at (r2) {};
  \node[subtree root] at (r3) {};
  \node[subtree root] at (X) {};

  \node[fat dot, red] at (OP) {};
  \node[fat dot, blue] at (OQ) {};
  \node[fat dot, red] at (p) {};
  \node[fat dot, blue] at (q) {};

  \node at ([xshift=3pt, yshift=5pt]X) {$x$};
  \node at ([xshift=3pt, yshift=5pt]r2) {$y_p$};
  \node at ([xshift=3pt, yshift=-6pt]r0) {$y_q$};

  \node at ([xshift=-8pt, yshift=0pt]p) {$p$};
  \node at ([xshift=-8pt, yshift=0pt]q) {$q$};
  \node at ([xshift=-20pt, yshift=0pt]OP) {$\orig(p)$};
  \node at ([xshift=-20pt, yshift=0pt]OQ) {$\orig(q)$};

  \node at ([xshift=+10pt, yshift=0pt]TU) {$\mathcal{T}$};
  \node at ([xshift=0pt, yshift=10pt]DR) {$\mathcal{D}$};

  \tikzset{alignment/.style={densely dotted, semithick}}

  \def\bracespace{1mm}

  \draw[<->, transform canvas={xshift=-\treespace}] (wDL) -- (wUL) node[midway, left=0pt] {$m$};
  \draw[alignment] (wDL) -- ([xshift=-\treespace]wDL);
  \draw[alignment] (wUL) -- ([xshift=-\treespace]wUL);

  \draw[<->, transform canvas={yshift=-\treespace}] (wDL) -- (wDR) node[midway, below=0pt] {$n$};
  \draw[alignment] (wDL) -- ([yshift=-\treespace]wDL);
  \draw[alignment] (wDR) -- ([yshift=-\treespace]wDR);

  \draw[<->, transform canvas={xshift=-\treespace}] (TL) -- (TU -| TL) node[midway, left=0pt] {$H$};
  \draw[alignment] (TL) -- ([xshift=-\treespace]TL);
  \draw[alignment] (TU) -- ([xshift=-\treespace]TU -| TL);

  \draw[<->, transform canvas={yshift=-\treespace}] (wDL -| DD) -- (wDL -| DR) node[midway, below=0pt] {$\kappa$};
  \draw[alignment] (DD) -- ([yshift=-\treespace]wDL -| DD);
  \draw[alignment] (DR) -- ([yshift=-\treespace]wDL -| DR);

  \end{tikzpicture}
  \caption{The query algorithm. The red and blue paths in $\Tcol$ mark the nodes in which $p$ and $q$ are active, respectively.
	The depths of the two origins correspond to two leaves of $\Ddep$, whose lowest common ancestor is~$x$; the children $y_p$ and $y_q$ of $x$ contain these two depths. For clarity, $\Ddep$ is drawn with its root on the right, so that the left-to-right order of the children of $x$ appears top-to-bottom.}
  \label{fig:query}
\end{figure}

The two origin depths determine two leaves of $\Ddep$.
Let $x\in\Ddep$ be the lowest common ancestor of these two leaves, and let $y_p$ and $y_q$ be the children of $x$ containing $\depth(\orig(p))$ and $\depth(\orig(q))$, respectively; see \Cref{fig:query} for an illustration.
Since the child intervals of $x$ are consecutive depth ranges, the origin that lies in the rightmost of these two children is the deeper one.
We stress that $x$ remains fixed for the rest of the query: the target sets $\actT{y}{u}$ are defined so that the origin depths of the two current points never leave $\interval{x}$, and hence the node $x$ never needs to be recomputed.

Suppose first that $\depth(\orig(q)) > \depth(\orig(p))$.
Equivalently, $y_q$ lies strictly to the right of $y_p$.
Let $u \coloneqq \orig(q)$.
By \Cref{lem:coactive-char}, we have $p \in \act{u}$, and since $\depth(\orig(p))$ lies in a child interval of $x$ strictly to the left of $y_q$, by the definition of $\actT{y_q}{u}$, we have $p\in \actT{y_q}{u}$.

We now use the lifting structure stored for $(y_q,u)$ to obtain $q'$, the predecessor of $q$ in $\actT{y_q}{u}$.
Since $p\in \actT{y_q}{u}$, \Cref{obs:equiv} tells us that $\weight(p)\le\weight(q')$ is equivalent to $\weight(p)\le\weight(q)$, and hence to $\mathcal{H}$.
At this point, three cases may occur.
\begin{enumerate}[(i)]
	\item If $q'=\Bot$, then the current hypothesis is false, and the query is resolved immediately.
	
	\item If $q'\neq\Bot$ and the child $y_{q'}$ of $x$ containing $\depth(\orig(q'))$ is different from $y_p$, then we replace $q$ with $q'$ and continue.
	Since $p, q' \in \act{u}$, the new pair $(p,q')$ is again co-active.
	
	\item If $q'\neq\Bot$ and $y_{q'}=y_p$, then this is the last lifting step, and it remains to compare $p$ and $q'$.
	Let $\interval{y_p}=[a\dd b]$.
	Since $q'$ was obtained from the lifting structure stored at $u=\orig(q)$, both $p$ and $q'$ are active in $u$.
	Moreover, $\depth(u)\in \interval{y_q}$ and $y_q$ lies strictly to the right of $y_p$, so $\depth(u)>b$.
	Let $v$ be the ancestor of $u$ at depth $b$.
	By \Cref{lem:active-path}, both $p$ and $q'$ are active in $v$.

	The origin depth of $p$ belongs to $\interval{y_p}$ by definition of $y_p$, and the origin depth of $q'$ belongs to $\interval{y_p}$ because $y_{q'}=y_p$.
	Since $\interval{y_p}$ is a canonical interval ending at $b$, we have $\interval{y_p}\subseteq C_b$.
	Hence both points belong to $\actL{v}$, and the ranking structure stored for $v$ compares them in constant time.
\end{enumerate}

The case in which $\depth(\orig(p)) > \depth(\orig(q))$, or equivalently, when $y_p$ lies strictly to the right of $y_q$, is symmetric: we take $u \coloneqq \orig(p)$, and use the lifting structure stored for $(y_p,u)$ to replace $p$ with its \emph{successor} $p'$ in $\actT{y_p}{u}$, where $q \in \actT{y_p}{u}$ by a symmetric argument.
As before, \Cref{obs:equiv} preserves the invariant, and the three possible outcomes mirror those of the previous case.
If $p' = \Top$, the hypothesis is false and the query is resolved immediately.
Otherwise, if the child $y_{p'}$ of $x$ containing $\depth(\orig(p'))$ differs from $y_q$, we continue with the new co-active pair $(p', q)$.
Finally, if $y_{p'} = y_q$ and $\interval{y_q} = [a \dd b]$, we use the ranking structure stored for the ancestor of $u$ at depth $b$ and compare $p'$ and $q$ in constant time.

Each nonterminal lifting step preserves the truth of the current hypothesis and strictly decreases the child index of the deeper of the two current origins among the children of the fixed node $x\in\Ddep$ (even though which of the two origins is deeper may alternate between steps).
Since $x$ has at most $\tau$ children, at most $\tau-1$ such steps are possible before the algorithm terminates.
Every step takes constant time, so the total query time is $\Oh(\tau)=\Oh(\log^{1/\kappa}n)$.

\paragraph{Space analysis.}
As in the baseline structure, the bulk of the argument relies on \Cref{lem:baseline-entropy}.
The new complication is that the families of sets are no longer globally pairwise disjoint.
To recover disjointness, we partition the stored structures into $\Oh(\kappa)$ groups, one per level of $\Ddep$, so that within each group the relevant sets are pairwise disjoint.
\Cref{lem:baseline-entropy} then applies separately on each group, and summing over the $\Oh(\kappa)$ levels of $\Ddep$ yields the claimed bound.

We begin with the local ranking structures.
For every node $u\in\Tcol$, the ranking structure stored for $u$ is built for the set $\actL{u}$, and its per-node cost, given by \Cref{lem:rank-primitive}, coincides with the one assumed in \Cref{lem:baseline-entropy} (with $X_u=\actL{u}$).
Recall that $\actL{u}$ consists of the points in $\act{u}$ whose origin depth falls in $C_d$, the largest canonical interval of $\Ddep$ ending at $d=\depth(u)$.

We group the stored ranking structures by the level of $\Ddep$ at which $C_d$ appears, that is, by the level of the unique $\Ddep$-node whose span is $C_d$.
Fix one level $\ell$, and consider all nodes $u\in\Tcol$ for which $C_{\depth(u)}$ appears at level $\ell$.
We claim that the corresponding sets $\actL{u}$ are pairwise disjoint.
Indeed, suppose that a point $p$ belonged to both $\actL{u}$ and $\actL{v}$.
Then $\depth(\orig(p))$ would lie in $C_{\depth(u)}\cap C_{\depth(v)}$, and since canonical intervals appearing at the same level of $\Ddep$ are pairwise disjoint, this forces $C_{\depth(u)}=C_{\depth(v)}$, and in turn $\depth(u)=\depth(v)$ because each $C_d$ ends at $d$.
Now $p$ is active in both $u$ and $v$, with the same depth in $\Tcol$, so \Cref{lem:active-path} gives $u=v$, proving the claim.

\Cref{lem:baseline-entropy} therefore applies and bounds the total space of the ranking structures associated with level $\ell$ by $\Oh(mn(\log m+\log\log n))$ bits.
Summing over the $\Oh(\kappa)$ levels of $\Ddep$, all ranking structures together use
$\Oh(\kappa mn(\log m+\log\log n))$ bits.

We now turn to the lifting structures.
Recall that for every non-leftmost child $y$ of a node $x\in\Ddep$, and every node $u\in\Tcol$ with $\depth(u)\in\interval{y}$, we store one lifting structure with source set $\acteq{u}$ and target set $\actT{y}{u}$.
The per-node space bound of \Cref{lem:lift-primitive} depends only on the source set and coincides with the one assumed in \Cref{lem:baseline-entropy} (with $X_u=\acteq{u}$).

Here the situation differs from the ranking case.
The source sets $\set{\acteq{u} : u\in\Tcol}$ already partition the set of all points; what prevents a direct application of \Cref{lem:baseline-entropy} is that each $\acteq{u}$ is reused as the source set of several stored lifting structures.
We claim that each $\acteq{u}$ appears as a source set at most $\kappa$ times.
Indeed, $\acteq{u}$ is the source set of the structure stored for a pair $(y,u)$ only when $y$ is a node of $\Ddep$ with $\depth(u)\in\interval{y}$, and there are $\kappa$ such nodes, as each is an ancestor of the leaf corresponding to $\depth(u)$.

We group the stored lifting structures by the level of $\Ddep$ containing $y$.
By the observation above, each $\acteq{u}$ appears as the source set of at most one structure per group, so the source sets within a group are pairwise disjoint.
\Cref{lem:baseline-entropy} therefore bounds the total space of the group by $\Oh(mn(\log m+\log\log n))$ bits.
Summing over the $\Oh(\kappa)$ levels of $\Ddep$, all lifting structures together use
$\Oh(\kappa mn(\log m+\log\log n))$ bits.

As in the baseline structure, storing $\depth(\orig(p))$ for every point uses $\Oh(mn\log\log n)$ bits, and the reduction structure of \Cref{lem:two-cand} uses $\Oh(mn\log m)$ bits.
Hence the entire structure takes
$\Oh(\kappa mn(\log m+\log\log n))$
bits, as promised.
This concludes the proof of \Cref{thm:main}.

\section{Construction of Local Primitives}\label{sec:implementation}

We now describe the construction of the two local primitives stated in
\Cref{lem:rank-primitive,lem:lift-primitive}.
Both constructions are entirely local to one fixed node of $\Tcol$, and rely on a common decomposition of the active set into at most $4m$ `monotone' subsets, called \emph{quarter-rows}.
They also make use of the two succinct encodings for sparse sets and monotone sequences introduced in \Cref{lem:sparse,lem:mono}.
The lifting structure uses the ranking structure as one of its sub-components, so we present the ranking primitive first and then build the lifting primitive on top of it.

\paragraph{Local coordinates.}
Throughout this section we fix one node $u\in\Tcol$ and denote
\[
	\interval{u}=[a\dd b]
	\qquad\text{and}\qquad
	\lambda \coloneqq |\interval{u}| = b-a+1.
\]
For every point $p\in [1\dd m]\times \interval{u}$, we define its \emph{local column in $u$} by
\[
	\col_u(p)\coloneqq \col(p)-a+1 \in [1\dd \lambda].
\]
From this point onward, we will view columns through their local coordinates inside $\interval{u}$.

We assume that the 1D RMQ structures of \Cref{lem:rmq1d} have been built for all rows, and we will use them only to test, in constant time, whether a point is visible from a given column.

\paragraph{Quarter-rows.}
We describe a canonical decomposition of $\act{u}$ into \emph{quarter-rows} that underlies both local structures.
If $u$ is a leaf, $\act{u}=[1\dd m]\times\{j\}$ for a single column $j$, and we define the $m$ quarter-rows to be the singletons $\{(i,j)\}$ for $i\in[1\dd m]$; all claims below are immediate in this case, so assume $u$ is internal, with children covering the intervals $[a_1\dd b_1]$ and $[a_2\dd b_2]$.

For any $p\in\act{u}$, the column $\col(p)$ belongs to exactly one of these two child~intervals.
Since $p$ is active in $u$, it is visible from some column in $\partial u = \{a_1,b_1,a_2,b_2\}$, and by monotonicity of visibility along its row, $p$ is also visible from at least one endpoint of the child interval that contains $\col(p)$.
We assign $p$ a \emph{canonical endpoint}: the left endpoint of the child interval containing $\col(p)$ if $p$ is visible from it, and its right endpoint otherwise.
For each pair~$(i,c) \in[1\dd m] \times \{a_1,b_1,a_2,b_2\}$, the corresponding \emph{quarter-row} is the set of all points in $\act{u}$ in row $i$ whose canonical endpoint is $c$;
we write $Q_u(p)$ for the quarter-row containing~$p$.

This partitions $\act{u}$ into at most $4m$ quarter-rows.
Given $p$, the quarter-row $Q_u(p)$ can be identified using $\Oh(\log m)$ bits in $\Oh(1)$ time by locating the child interval containing $\col(p)$ and testing visibility from its endpoints using \Cref{lem:rmq1d}.

\begin{observation}\label{obs:quarter-row-monotone}
	Let $Q$ be a quarter-row, and list its points in increasing order of weight.
	Then their local columns form a strictly monotone sequence.
	More precisely: 
	\begin{enumerate}[(i)]
		\item if $Q$ is defined by the left endpoint of a child interval, the sequence is strictly increasing;
		\item if $Q$ is defined by the right endpoint of a child interval, the sequence is strictly decreasing.
	\end{enumerate}
\end{observation}
\begin{proof}
	The leaf case is trivial, so assume that $u$ is internal.
	First suppose that $Q$ is defined by the left endpoint $c$ of one of the two child intervals of $u$.
	Let $p,q\in Q$ with $\col(p)<\col(q)$.
	Since both points lie in the same row and $q$ is visible from $c$, the point $q$ has the largest weight on the interval of that row between $c$ and $\col(q)$.
	This interval contains $\col(p)$, so $\weight(p)<\weight(q)$.
	Hence the weight is a strictly increasing function of the column.

	Now suppose that $Q$ is defined by the right endpoint $c$ of a child interval.
	Let again $p,q\in Q$ with $\col(p)<\col(q)$.
	Since $p$ is visible from $c$, the point $p$ has the largest weight on the interval of that row between $\col(p)$ and $c$.
	This interval contains $\col(q)$, so $\weight(p)>\weight(q)$.
	Hence the weight is a strictly decreasing function of the column.
\end{proof}

\begin{lemma}\label{claim:quarter-row-subset}
	Let $Q$ be a quarter-row and let $X\subseteq Q$.
	There is an encoding using
	\[
		\Oh\!\left(|X|+|X|\log\frac{\lambda}{|X|}\right)
	\]
	bits that, given a point $p\in X$, returns in $\Oh(1)$ time the rank of $p$ in $X$.
\end{lemma}
\begin{proof}
	Store the set
	$C_X \coloneqq \{\col_u(p): p\in X\}\subseteq [1\dd \lambda]$
	using \Cref{lem:sparse}.
	This takes
	\[
		\Oh\!\left(|X|+|X|\log\frac{\lambda}{|X|}\right)
	\]
	bits and supports rank queries on local columns in constant time.

	By \Cref{obs:quarter-row-monotone}, inside $Q$ the weight is a strictly monotone function of the column.
	Hence, once we know the rank of $\col_u(p)$ in the sorted set $C_X$, we also know the rank of $p$ in~$X$.
	More explicitly, if the quarter-row is of the increasing type, then the two ranks coincide, whereas in the decreasing case they sum to $|X|+1$.
	In both cases the answer is obtained in $\Oh(1)$ time.
\end{proof}

We are now ready to prove the two local primitives of \Cref{lem:rank-primitive,lem:lift-primitive}, starting with \Cref{lem:rank-primitive}, which we restate below.

\RestateRankPrimitive*

\begin{proof}
	We partition $L$ according to the quarter-rows of $u$.
	For every quarter-row $Q$, we define
	$L_Q \coloneqq L\cap Q$.
	Let $\mathcal{Q}$ denote the family of quarter-rows $Q$ with $L_Q \neq \emptyset$.

	The structure has two layers.
	The first layer handles the order inside one quarter-row, and the second layer handles how the quarter-rows are interleaved in the global order of $L$.

	\paragraph{First layer.}
	For every nonempty $L_Q$, we build the ranking structure of \Cref{claim:quarter-row-subset}.
	Since the sets $L_Q$ partition $L$, \Cref{lem:logsum} with $t = |\mathcal{Q}| \le 4m$ bounds their total space by
	\[
		\sum_{Q \in \mathcal{Q}} \Oh\!\left(|L_Q|+|L_Q|\log\frac{\lambda}{|L_Q|}\right) \le \Oh\!\left(|L|+|L|\log\frac{m\lambda}{|L|}\right).
	\]

	\paragraph{Second layer.}
	Let $p_1,p_2,\dots,p_{|L|}$ be the points of $L$ listed in increasing order of weight.
	For every $Q \in \mathcal{Q}$, let $B_Q \coloneqq \{i\in[1\dd |L|] : p_i\in L_Q\}$ be the set of positions occupied by points of $L_Q$ in this sequence.
	We encode each $B_Q$ using \Cref{lem:sparse}, supporting select in $\Oh(1)$ time.
	Since the sets $L_Q$ partition $L$, \Cref{lem:logsum} with $t = |\mathcal{Q}| \le 4m$ bounds their total space by
	\[
		\sum_{Q \in \mathcal{Q}} \Oh\!\left(|L_Q|+|L_Q|\log\frac{|L|}{|L_Q|}\right)
		\le \Oh\!\left(|L|+|L|\log m\right) = \Oh(|L|\log m).
	\]

	\paragraph{Query algorithm.}
	Let $p\in L$ be the query point.
	We first determine its quarter-row $Q \coloneqq Q_u(p)$ in $\Oh(1)$ time.
	Using the structure stored for $L_Q$, we compute in $\Oh(1)$ time the rank $r$ of $p$ in $L_Q$.
	Since $p$ is the $r$-th lightest point of $L_Q$, its rank in $L$ equals $\select_{B_Q}(r)$, which we compute in $\Oh(1)$ time.
\end{proof}

We now proceed with the proof of \Cref{lem:lift-primitive}, restated below.

\RestateLiftPrimitive*

\begin{proof}
	We describe the structure for \emph{successor} queries; the predecessor case is symmetric.
	The construction has two layers: the first determines which quarter-row contains the successor of the query point, and the second locates the successor within that quarter-row.

	\paragraph{First layer.}
	We build a local ranking structure of \Cref{lem:rank-primitive} for the set $S$, using $\Oh(|S|\log m + |S|\log(m\lambda/|S|))$ bits.
	We also store an array of length $|S|$, where the $i$\nobreakdash-th entry is either the identifier of the quarter-row containing the successor of $s_i$ in $T$, or $\Top$ if no such successor exists, where $s_1,\dots,s_{|S|}$ are the points of $S$ listed in increasing order of weight.
	Since each identifier uses $\Oh(\log m)$ bits, the array uses $\Oh(|S|\log m)$ bits in total.

	\paragraph{Second layer.}
	For every quarter-row $Q$, we define
	\[
		S_Q \coloneqq \set{p\in S : \text{the successor of $p$ in $T$ belongs to } Q}.
	\]
	Let $\mathcal{Q}$ be the family of quarter-rows $Q$ with $S_Q \neq \emptyset$.
	For every $Q \in \mathcal{Q}$, we store two objects.

	First, we build the ranking structure of \Cref{lem:rank-primitive} for the set $S_Q$. All these structures~use
	\[
		\sum_{Q\in\mathcal{Q}}
		\Oh\!\left(|S_Q|\log m+|S_Q|\log\frac{m\lambda}{|S_Q|}\right)
	\]
	bits of space in total. The first summand contributes $\sum_{Q\in\mathcal{Q}} \Oh(|S_Q|\log m) \le \Oh(|S|\log m)$, since the sets $S_Q$ are pairwise disjoint subsets of $S$.
	For the second, \Cref{lem:logsum} gives
	\[
		\sum_{Q\in\mathcal{Q}} \Oh\!\left(|S_Q|\log\frac{m\lambda}{|S_Q|}\right)
		\le \Oh\!\left(|S|\log\frac{m^2\lambda}{|S|}\right)
		= \Oh\!\left(|S|\log m+|S|\log\frac{m\lambda}{|S|}\right).
	\]

	Second, let $p_1,\dots,p_{|S_Q|}$ be the points of $S_Q$ in increasing order of weight, and let $q_i$ be the successor of $p_i$ in $T$.
	We store the sequence $c_Q[i] \coloneqq \col_u(q_i)$ of their local columns.

	The sequence $c_Q$ is monotone: if $p_i$ is lighter than $p_j$, then $q_i$ is no heavier than $q_j$, and since all $q_i$ belong to the same quarter-row $Q$, their local columns form a monotone sequence, by \Cref{obs:quarter-row-monotone}.
	Thus, by \Cref{lem:mono}, each $c_Q$ can be stored in $\Oh(|S_Q|+|S_Q|\log(\lambda/|S_Q|))$ bits with $\Oh(1)$-time access, and \Cref{lem:logsum} bounds their total space~by
	\[
		\sum_{Q\in\mathcal{Q}} \Oh\!\left(|S_Q|+|S_Q|\log\frac{\lambda}{|S_Q|}\right)
		\le \Oh\!\left(|S|+|S|\log\frac{m\lambda}{|S|}\right).
	\]

	\paragraph{Query algorithm.}
	Let $p\in S$ be the query point.
	We compute the rank $r$ of $p$ in $S$ and inspect the $r$\nobreakdash-th entry of the array stored in the first layer.
	If this entry is $\Top$, then $p$ has no successor in $T$.
	Otherwise, it identifies a quarter-row $Q$ containing the successor of $p$.
	Since $p\in S_Q$, we query the local ranking structure for $S_Q$ to obtain the rank $r_Q$ of $p$ in $S_Q$, and then extract the local column $c_Q[r_Q]$ of the successor.
	Since the identifier of $Q$ determines the row of the successor, we reconstruct the successor point in constant time.
\end{proof}

\bibliographystyle{plainurl}
\bibliography{biblio}

\end{document}